\documentclass[conference]{IEEEtran}
\IEEEoverridecommandlockouts

\usepackage{aa-sty-succ-rep}

\begin{document}

\title{Reducing NEXP-complete problems to DQBF\thanks{To appear in the proceedings of FMCAD 2022.}}

\author{
\IEEEauthorblockN{Fa-Hsun Chen}
\IEEEauthorblockA{\textit{National Taiwan University}\\
r10944015@ntu.edu.tw}
\and
\IEEEauthorblockN{Shen-Chang Huang}
\IEEEauthorblockA{\textit{National Taiwan University}\\
b07902135@ntu.edu.tw}
\and
\IEEEauthorblockN{ Yu-Cheng Lu}
\IEEEauthorblockA{\textit{National Taiwan University}\\
luyucheng@protonmail.com}
\and
\IEEEauthorblockN{Tony Tan}
\IEEEauthorblockA{\textit{National Taiwan University}\\
tonytan@csie.ntu.edu.tw}
}

\maketitle

\begin{abstract}
We present an alternative proof of the NEXP-hardness of the satisfiability of {\em Dependency Quantified Boolean Formulas} (DQBF).
Besides being simple, our proof also gives us a general method to reduce NEXP-complete problems to DQBF.
We demonstrate its utility by presenting explicit reductions
from a wide variety of NEXP-complete problems to DQBF such as (succinctly represented) 3-colorability,
Hamiltonian cycle, set packing and subset-sum
as well as NEXP-complete logics such as 
the Bernays-Sch\"onfinkel-Ramsey class, the two-variable logic and the monadic class. 
Our results show the vast applications of DQBF solvers which recently have gathered 
a lot of attention among researchers.
\end{abstract}



\section{Introduction}
\label{sec:intro}

The last few decades have seen a tremendous development of boolean SAT solvers
and their applications in many areas of computing~\cite{sat-handbook}.
Motivated by applications in verification and synthesis of hardware/software
designs~\cite{Jiang09,BalabanovJ15,SchollB01,GitinaRSWSB13,BloemKS14,ChatterjeeHOP13,KuehlmannPKG02},
researchers have recently looked at the generalization of boolean formulas known as 
{\em dependency quantified boolean formulas} (DQBF).

While solving boolean SAT is ``only'' $\npt$-complete,
for DQBF the complexity jumps to $\nexpt$-complete~\cite{PetersonR79}.
This makes solving DQBF quite a challenging research topic.
Nevertheless there has been exciting progress.
See, e.g.,~\cite{BalabanovCJ14,FrohlichKB12,Ge-ErnstSW19,KullmannS19,WimmerSB19,WimmerWSB16,WimmerRM017,Kovasznai16,SchollW18}
and the references within, as well as
solvers such as iDQ~\cite{FrohlichKBV14}, dCAQE~\cite{TentrupR19}, HQS~\cite{GitinaWRSSB15,WimmerKBS017} and DQBDD~\cite{SicS21}.
A natural question to ask is if we can use DQBF solvers to solve any $\nexpt$-complete problems
-- similar to how SAT solvers are used to solve any $\npt$-complete problems.

%

In this short paper we show how to reduce 
a wide variety of $\nexpt$-complete problems to DQBF, especially the succinctly represented problems that 
recently have found applications in hardware/software engineering~\cite{avi-succ-83,KiniM018,PavlogiannisSSC20}.
We present another proof for the $\nexpt$-hardness of DQBF.
We actually give two proofs.
The first is by a very simple reduction from {\it succinct 3-colorability}~\cite{yanna-succ-86}.
The second is by utilizing the notion that we call {\em succinct projection}.
It is the second one that we view more interesting since it gives us a general method to reduce
any $\nexpt$-complete problem to DQBF.

The main idea is quite standard: 
We encode the accepting runs of a non-deterministic Turing machine (with exponential run time)
with boolean functions of polynomial arities.
However, we observe that the input-output relation of these functions can 
actually be ``described'' by small circuits/formulas.
Succinct projections are simply deterministic algorithms that construct these circuits efficiently.
This simple observation is a deviation from the standard definition of $\nexpt$, that
a language in $\nexpt$ is a language with an exponentially long certificate.

Using succinct projections, we present reductions from
various $\nexpt$-complete problems such as
(succinct) {\em Hamiltonian cycle}, {\em set packing} and {\em subset sum}.
We believe our technique can be easily modified for many other natural problems.
Note that the reduction in~\cite{PetersonR79} gives little insight on how it can be used 
to obtain explicit reductions from concrete $\nexpt$-complete problems.

We also present the reductions from well known $\nexpt$-complete logics
such as {\em the Bernays-Sch\"onfinkel-Ramsey class}, {\em two-variable logic} ($\fotwo$)
and {\em the L\"owenheim class}~\cite{Lewis80,Furer83,gkv,bgg97,llt21}.
In fact we show that they are essentially equivalent to DQBF.
Note that these are logics that have found applications in AI~\cite{dl-handbook03}, 
databases~\cite{kotek-paper} and automated reasoning~\cite{aut-reasoning},
but lack implementable algorithms.
Prior to our work, the only algorithm known for these logics is to 
``guess'' a model (of exponential size) and then verify that it is indeed a model of the input formula.
Recent work in~\cite{kotek-paper} reduces $\fotwo$ formulas to exponentially long SAT instances,
but the experimental results are not promising.

We hope that the technique introduced in this short paper
can lead to richer applications of DQBF solvers as well as a wide variety of benchmarks
which in turn can lead to further development.
It is also open whether the class $\nexpt$ has a {\em bona-fide} problem~\cite{yanna-succ-86}.
Our paper demonstrates that DQBF can be a good candidate -- akin to how boolean SAT
is the central problem in the class $\npt$.

This paper is organized as follows.
In Sect.~\ref{sec:prelim} we review some definitions and terminology.
In Sect.~\ref{sec:succ-proj} we reprove the $\nexpt$-completeness of solving DQBF.
In Sect.~\ref{sec:reduction} and~\ref{sec:logic}
we present concrete reductions from some $\nexpt$-complete problems and logics to DQBF instances.
Missing details can be found in the appendix.


\section{Preliminaries}
\label{sec:prelim}

Let $\Sigma=\{0,1\}$.
We usually use the symbol $\va,\vb,\vc$ (possibly indexed) to denote a string in $\Sigma^*$
with $|\va|$ denoting the length of $\va$.
We use $\vx,\vy,\vz,\vu,\vv$ to denote vectors of boolean variables.
The length of $\vx$ is denoted by $|\vx|$.
We write $C(\vu)$ to denote a (boolean) circuit $C$ with input gates $\vu$.
When the input gates are not relevant or clear from the context, 
we simply write $C$.
For $\va\in \Sigma^{|\vu|}$, $C(\va)$ denotes the value of $C$ 
when we assign the input gates $\vu$ with $\va$.
All logarithms have base $2$.

A {\em dependency quantified boolean formula} (DQBF) in prenex normal form is a formula of the form:
\begin{align}
\label{eq:dqbf}
\Psi & := 
\forall x_1 \cdots \ \forall x_{n}\
\exists y_1(\vz_1) \cdots \ \exists y_m(\vz_m) \quad \psi
\end{align}
where each $\vz_i$ is a vector of variables from $\{x_1,\ldots,x_n\}$
and $\psi$, called {\em the matrix}, is a quantifier-free boolean formula using variables $x_1,\ldots,x_n,y_1,\ldots,y_m$. 
The variables $x_1,\ldots,x_n$ are called {\em the universal variables},
$y_1,\ldots,y_m$ {\em the existential variables}
and each $\vz_i$ {\em the dependency set} of $y_i$.


A DQBF $\Psi$ in the form~(\ref{eq:dqbf}) is {\em satisfiable}, 
if for every $1\leq i \leq m$,
there is a function $s_i:\Sigma^{|\vz_i|}\to\Sigma$ such that
by replacing each $y_i$ with $s_i(\vz_i)$,
the formula $\psi$ becomes a tautology.
The function $s_i$ is called the {\em Skolem function} for $y_i$.
In this case, we also say that $\Psi$ is satisfiable by the Skolem functions $s_1,\ldots,s_m$.
The problem $\satdqbf$ is defined as:
On input DQBF $\Psi$ in the form~(\ref{eq:dqbf}),
decide if it is satisfiable.

Since many $\nexpt$-complete problems use circuits as the succinct representations of the inputs,
we allow the matrix $\psi$ to be in {\em circuit form},
i.e., $\psi$ is given as a (boolean) circuit with input gates $x_1,\ldots,x_n,y_1,\ldots,y_m$.
This does not effect the generality of our results, since
every DQBF in circuit form can be converted to one in the standard formula form
as stated in Proposition~\ref{prop:dnf}.

\begin{proposition}
\label{prop:dnf}
Every DQBF $\Psi$ in the form of (\ref{eq:dqbf}) in circuit form
can be converted in polynomial time into an equisatisfiable DQBF formula $\Psi'$ whose matrix is in DNF.
Moreover, $\Psi$ and $\Psi'$ have the same existential variables (with the same dependency set).
\end{proposition}

The proof is by standard Tseitin's transformation~\cite{tseitin}.
As an example, consider the following DQBF.
\begin{align*}
 &  \forall x_1 \forall x_2\ \exists y_1(x_1) \exists y_2(x_2)
\ \neg \big(x_2 \vee (y_1\wedge x_1 \wedge y_2)\big)
\end{align*}
It is equisatisfiable with the following DQBF.
\begin{align*}
&  \forall x_1 \forall x_2\ \forall u_1\forall u_2\forall u_3\ \forall v_1\forall v_2\ \exists y_1(x_1) \exists y_2(x_2)
\\
& 
\left(
\begin{array}{l}
(v_1\leftrightarrow y_1)\wedge (v_2\leftrightarrow y_2)
\wedge 
(u_1\leftrightarrow v_1\wedge x_1\wedge v_2)
\\
\wedge
(u_2\leftrightarrow x_2\vee u_1)
\wedge
(u_3\leftrightarrow \neg u_2)
\end{array}
\right)
\to u_3
\end{align*}
Intuitively, we use the extra variable $v_1$ to represent the value $y_1$,
$v_2$ the value $y_2$, $u_1$ the value $y_1\wedge x_1\wedge y_2$,
$u_2$ the value $x_2 \vee (y_1\wedge x_1\wedge y_2)$
and $u_3$ the value $\neg (x_2\vee (y_1\wedge x_1\wedge y_2))$.
Note that the matrix can be easily rewritten into DNF.

\section{The $\nexpt$-completeness of $\satdqbf$}
\label{sec:succ-proj}

In this section we present two new proofs
that $\satdqbf$ is $\nexpt$-complete, originally proved in~\cite{PetersonR79}.

\begin{theorem}{\em\cite{PetersonR79}}
$\satdqbf$ is $\nexpt$-complete.
\end{theorem}

Note that the membership is straightforward.
So we will focus only on the hardness.

\subsection{The first proof: Reduction from succinct 3-colorability}
\label{subsec:first-proof}

The reduction is from the problem {\em graph 3-colorability}
where the input graphs are given in a succinct form~\cite{avi-succ-83}.
A (boolean) circuit $C(\vu,\vv)$, where $|\vu|=|\vv|=n$, represents a graph $G(C)=(V,E)$ where $V=\Sigma^n$
and $(\va,\vb)\in E$ iff $C(\va,\vb)=1$.
The problem {\em succinct 3-colorability} is defined as:
On input circuit $C$, decide if $G(C)$ is 3-colorable.
This problem is $\nexpt$-complete~\cite{yanna-succ-86}.

The reduction to $\satdqbf$ is as follows.
Let $C(\vu,\vv)$ be the input circuit, where $|\vu|=|\vv|=n$.
We represent a 3-coloring of $G(C)$ as a function $g:\Sigma^n\to \{01,10,11\}$
which can be encoded by the following DQBF.
\begin{align}
\nonumber
\Psi := & \forall \vx_1 \forall \vx_2 \ \exists y_{1}(\vx_1)\exists y_{2}(\vx_1)
\ \exists y_{3}(\vx_2)\exists y_{4}(\vx_2)
\\
&\hspace{1.25cm} \vx_1=\vx_2 \ \to\ (y_1,y_2)=(y_3,y_4)
\\
&\hspace{.75cm} \wedge\
(y_1,y_2)\neq (0,0)\ \wedge\ (y_3,y_4)\neq(0,0)
\\
&\hspace{.75cm} \wedge\
C(\vx_1,\vx_2)=1 \ \to \ (y_1,y_2)\neq (y_3,y_4) 
\end{align}
Intuitively, we use $y_1,y_2$ and $y_3,y_4$ to represent the first and the second bits of the image $g(\vx_1)$ and $g(\vx_2)$, respectively.
Lines (2) and (3) state that $(y_1,y_2)$ and $(y_3,y_4)$ must represent the same function from $\Sigma^n$ to $\Sigma^2$
and that their images do not inclue $00$.
Line (4) states that the colors of two adjacent vertices must be different.
Thus, $G(C)$ is 3-colorable iff $\Psi$ is satisfiable.

\subsection{The second proof: Reduction via succinct projections}
\label{subsec:second-proof}

Our second proof uses the notion of {\em succinct projection}.
We need some terminology.
Let $C(\vu_1,\vv_1,\vu_2,\vv_2)$ be a circuit with input gates $\vu_1,\vv_1,\vu_2,\vv_2$
where $|\vu_1|=|\vu_2|=n$ and $|\vv_1|=|\vv_2|=m$.
We say that a function $g:\Sigma^n\to\Sigma^m$ {\em agrees} with the circuit $C$,
if $C(w_1,g(w_1),w_2,g(w_2))=1$, for every $w_1,w_2\in \Sigma^n$.
In this case, we also say that the circuit $C$ {\em describes} the function $g$.
In the following whenever we say that a function $g:\Sigma^n\to\Sigma^m$ agrees with $C(\vu_1,\vv_1,\vu_2,\vv_2)$,
we implicitly assume that $n=|\vu_1|=|\vu_2|$ and $m=|\vv_1|=|\vv_2|$.

\begin{definition}
\label{def:succ-proj}
A {\em succinct projection} for a language $L$ is a polynomial time deterministic algorithm~$\cM$
such that on input $w\in \Sigma^*$, $\cM$ outputs a circuit $C$ such that
 $w\in L$ iff there is a function $g$ that agrees with $C$.
\end{definition}

Intuitively, we can view the function $g$ as the certificate for the membership of $w$ in $L$
and the circuit $C$ as the succinct description of $g$.
Since succinct projection runs in polynomial time,
the output circuit can only have polynomially many gates.
The following theorem is a new characterization of languages in $\nexpt$.

\begin{theorem}
\label{theo:nexpt-succ-proj}
A language $L\in \nexpt$ iff it has a succinct projection.
\end{theorem}
\begin{proof}
(if) Suppose that $L$ has a succinct projection.
Consider the following algorithm.
On input $w$, first use the succinct projection to construct the circuit $C$.
Then, guess a function $g$ (of exponential size)
and verify that it agrees with $C$.
It is obvious that it runs in non-deterministic exponential time.
That it is correct follows from the definition of succinct projection.

(only if) It is essentially the Cook-Levin reduction disguised 
in the form of function certificates. 
We only sketch it here. 
Let $L\in \nexpt$ and $M$ be a 1-tape NTM that accepts $L$ in time $2^{p(n)}$
for some polynomial $p(n)$.
For a word $w\in L$ of length $n$, its accepting run can be represented
as a function $g:\Sigma^{p(n)}\times \Sigma^{p(n)}\to \Sigma^{\ell}$,
where $g(i,j)$ denotes the content of cell $i$ in time $j$.
The tuples in the codomain $\Sigma^{\ell}$ encode the states and the tape symbols of~$M$.
To verify that $g$ represents an accepting run,
it is sufficient to verify that for every $i_1,j_1,i_2,j_2 \in \Sigma^{p(n)}$,
the tuple $(i_1,j_1,g(i_1,j_1),i_2,j_2,g(i_2,j_2))$ satisfies a certain property $P$
which depends only on the input word $w$ and the transitions of $M$.
The desired succinct projection constructs in polynomial time a circuit $C$
describing this property $P$.
\end{proof}

\noindent
{\em The second proof of the $\nexpt$-hardness of $\satdqbf$:}
Let $L \in \nexpt$.
The polynomial time (Karp) reduction from $L$ to $\satdqbf$ 
is described as Algorithm~1 below.

\noindent
\begin{center}
\begin{tabular}{|p{0.46\textwidth}|}
 \hline
\multicolumn{1}{|c|}{\bf Algorithm~1: Reducing $L\in \nexpt$ to $\satdqbf$}
\\ 
\hline
\hline
{\bf Input:} $w\in \Sigma^*$.
\\
{\footnotesize 1:} Run the succinct projection of $L$ on $w$.
\\
{\footnotesize 2:} Let $C(\vx_1,\vy_1,\vx_2,\vy_2)$ be the output circuit where
\\
\hspace{0.25cm}
$|\vx_1|=|\vx_2|=n$, $|\vy_1|=|\vy_2|=m$, $\vy_1=(y_{1,1},\ldots,y_{1,m})$
\\
\hspace{0.25cm}
 and $\vy_2=(y_{2,1},\ldots,y_{2,m})$.
\\
{\footnotesize 3:} Output the following DQBF $\Psi$:
\\
\hspace{0.45cm}$\forall \vx_1 \forall \vx_2 
\ \exists y_{1,1}(\vx_1)\cdots \exists y_{1,m} (\vx_1)
\ \exists y_{2,1}(\vx_2)\cdots \exists y_{2,m} (\vx_2)$
\\
$\qquad\;\;\; \ C(\vx_1,\vy_1,\vx_2,\vy_2)\ \wedge \ \big(\vx_1= \vx_2 \to \vy_1= \vy_2\big)$
\\
\hline
\end{tabular}
\end{center}

We show $w\in L$ iff $\Psi$ is satisfiable.
Suppose $w\in L$.
Let $g:\Sigma^n\to\Sigma^m$ be a function that agrees with $C$.
For each $1\leq i\leq m$, define the Skolem function $s_i:\Sigma^n\to\Sigma$ 
where $s_i(\va)$ is the $i$-th component of $g(\va)$, for every $\va\in \Sigma^n$.
It is routine to verify that $\Psi$ is satisfiable 
with each $s_i$ being the Skolem function for $y_{1,i}$ and $y_{2,i}$.

Conversely, suppose $\Psi$ is satisfiable.
Let $s_{j,i}:\Sigma^n\to\Sigma$ be the Skolem function for $y_{j,i}$,
where $1\leq j\leq 2$ and $1\leq i\leq m$.
Since $\vx_1= \vx_2\ \to\ \vy_1=\vy_2$,
the functions $s_{1,i}$ and $s_{2,i}$ must be the same, for every $1\leq i \leq m$.
Define $g:\Sigma^n\to\Sigma^m$ where $g(\va)=(s_1(\va),\ldots,s_{1,m}(\va))$
for every $\va\in\Sigma^n$.
Since $C(\va_1,g(\va_1),\va_2,g(\va_2))$ is true for every $\va_1,\va_2$,
the function $g$ agrees with $C$.
That is, there is a function that agrees with $C$.
Hence, $w\in L$.
This completes the second proof.

\begin{remark}
\label{rem:succ-proj-np}
Observe that when Theorem~\ref{theo:nexpt-succ-proj} is applied to languages in $\npt$,
the accepting run of a non-deterministic Turing machine with polynomial run time $p(n)$
is represented as a function $g:\Sigma^{\log p(n)}\times\Sigma^{\log p(n)} \to \Sigma^{\ell}$
and the succinct projection outputs a circuit $C(\vx_1,\vy_1,\vx_2,\vy_2)$
where $|\vx_1|=|\vx_2|=\log p(n)$ and $|\vy_1|=|\vy_2|=\ell$.
Thus, for $L \in \npt$,
the DQBF output by Algorithm~1 has $4\log p(n)$ universal variables and $2\ell$ existential variables.
\end{remark}


\section{Some concrete reductions}
\label{sec:reduction}

In this section we show how to utilize succinct projection
to obtain the reductions from concrete $\nexpt$-complete problems to $\satdqbf$.
These are (succinct) {\em Hamiltonian cycle}, {\em set packing} and {\em subset sum}~\cite{yanna-succ-86}.
We use the notion of succinctness from~\cite{avi-succ-83}
which has been explained in Sect.~\ref{subsec:first-proof}.
By Algorithm~1, it suffices to present only the succinct projections.

\paragraph*{Some useful notations}
For an integer $k\geq 1$, $[k]$ denotes the set $\{0,\ldots,k-1\}$.
For $i \in [2^n]$, $\bin_n(i)$ is the binary representation of $i$ in $n$ bits.
The number represented by $\va\in \Sigma^n$ is denoted by $\num(\va)$.
For $\va,\vb\in \Sigma^n$,
if $\num(\va)=\num(\vb)+1 \pmod {2^n}$,
we say that {\em $\va$ is the successor of $\vb$}, denoted by $\va=\vb+1$.
Note that successor is applied only on two strings with the same length
and the successor of $1^n$ is $0^n$.
It is not difficult to construct a circuit $C(\vx,\vy)$ (in time polynomial in $|\vx|+|\vy|$)
such that $C(\va,\vb)=1$ iff $\va=\vb+1$.

\paragraph*{Reduction from succinct Hamiltonian cycle}
Succinct Hamiltonian cycle is defined as follows.
The input is a circuit $C(\vu,\vv)$.
The task is to decide if there is a Hamiltonian cycle in $G(C)$.

Let $C(\vu,\vv)$ be the input circuit where $|\vu|=|\vv|=n$.
We use a function $g:\Sigma^n\to\Sigma^n$
to represent a Hamiltonian cycle $(\vb_0,\ldots,\vb_{2^n-1})$ where $g(\bin_n(i))=\vb_i$, for every $i\in [2^n]$.
To correctly represent a Hamiltonian cycle,
the following must hold for every $\va_1,\va_2 \in \Sigma^n$.
\begin{enumerate}[(H1)]
\item
If $\va_1\neq \va_2$, then $g(\va_1)\neq g(\va_2)$.
\item
If $\va_2=\va_1+1$, then $(g(\va_1),g(\va_2))$ is an edge in $G(C)$.
\end{enumerate}
The succinct projection for succinct Hamiltonian cycle simply outputs the 
circuit that expresses (H1) and (H2), i.e.,
it outputs the following circuit $D(\vx_1,\vy_1,\vx_2,\vy_2)$ where $|\vx_1|=|\vx_2|=|\vy_1|=|\vy_2|=n$:
\begin{align*}
\big(\vx_1\neq \vx_2  \to  \vy_1\neq \vy_2\big)
& \wedge 
\big(\vx_2=\vx_1+1  \to  C(\vy_1,\vy_2)=1\big)
\end{align*}
Obviously, a function $g:\Sigma^n\to\Sigma^n$ represents a hamiltonian cycle in $G(C)$ iff it agrees with $D$.

\paragraph*{Reduction from succinct set packing}
In the standard representation the problem {\em set packing} is defined as follows.
The input is a collection $\cK$ of finite sets $S_1,\ldots,S_{\ell}\subseteq \Sigma^m$ and an integer $k$.
The task is to decide whether $\cK$ contains $k$ mutually disjoint sets.
We assume each $S_i$ has a ``name'' which is a string in $\Sigma^{\log \ell}$.

The succinct representation of the sets $S_1,\ldots,S_{\ell}$ is a circuit $C(\vu,\vv)$ where $|\vu|=m$ and $|\vv|=\log \ell$.
A string $\va\in \Sigma^m$ is in the set $S_{\vb}$, if $C(\va,\vb)=1$.
We denote by $\cK(C)$ the collection of finite sets defined by the circuit $C$.
The problem {\em succinct set packing} is defined analogously where the input is the circuit $C(\vu,\vv)$ and an integer $k$ (in binary).

We now describe its succinct projection.
Let $C(\vu,\vv)$ and $k$ be the input where $|\vu|=m$ and $|\vv|=n$.
We first assume that $k$ is a power of $2$.
We represent $k$ disjoint sets $S_1,\ldots,S_k$ in $\cK(C)$ as a function $g:\Sigma^{\log k}\times \Sigma^{m}\to \Sigma^n$
where $g(\bin(i),\va)$ is the name of the set $S_i$.
Note that the string $\va$ is actually ignored in the definition of $g$.

For a function $g:\Sigma^{\log k}\times \Sigma^{m}\to \Sigma^n$ to correctly represent $k$ disjoint sets,
the following must hold for every $(\va_1,\vb_1),(\va_2,\vb_2) \in \Sigma^{\log k}\times \Sigma^{m}$.
\begin{enumerate}[(P1)]
\item
If $\va_1=\va_2$, then $g(\va_1,\vb_1)=g(\va_2,\vb_2)$.
That is, the function $g$ does not depend on $\vb_1$ and $\vb_2$.
\item
If $\va_1\neq \va_2$ and $\vb_1=\vb_2$, then $C(\vb_1,g(\va_1,\vb_1))=0$ or $C(\vb_1,g(\va_2,\vb_2))=0$.
That is, the element $\vb_1$ is not in the sets whose names are $g(\va_1,\vb_1)$ and $g(\va_2,\vb_2)$.
\end{enumerate}
It is routine to verify that $g$ represents $k$ disjoint sets iff 
(P1) and (P2) hold for every $(\va_1,\vb_1),(\va_2,\vb_2) \in \Sigma^{\log k}\times \Sigma^{m}$.
The succinct projection outputs the following circuit $D$ that formalizes (P1) and (P2):
\begin{align*}
&  \big(\vx_1 = \vx_2  \to  \vz_1=\vz_2\big)
\\
\wedge &
\big(\vx_1\neq \vx_2 \wedge  \vy_1=\vy_2\big)  \to 
\neg \big(C(\vy_1,\vz_1)= C(\vy_1,\vz_2)=1\big)
\end{align*}
If $k$ is not a power of $2$,
we conjunct both atoms $\vx_1=\vx_2$ and $\vx_1\neq\vx_2$ with a circuit that tests whether 
the numbers represented by the bits $\vx_1$ and $\vx_2$ is an integer in $[k]$.
Such a circuit can be easily constructed in polynomial time in $\lceil\log k\rceil$.

\paragraph*{Reduction from succinct subset-sum}
In the standard representation the instance of subset-sum is a list of positive integers 
$s_0,\ldots,s_{k-1}$ and $t$ (all written in binary).
The task is to decide if there is a subset $X\subseteq [k]$ such that $\sum_{i \in X} s_i = t$.
Such $X$ is called the subset-sum solution.
The succinct representation is defined as 
two circuits $C_1(\vu_1,\vv)$ and $C_2(\vu_2)$, where $|\vu_1|=\max_{i\in [k]}\log s_i$,
$|\vv|=\log k$ and $|\vu_2|=\log t$.
Circuit $C_1$ defines the numbers $s_i$'s where $C_1(\va,\vb)$ is the $i$-th least significant bit of $s_j$,
where $i=\num(\va)$ and $j=\num(\vb)$.
Circuit $C_2$ defines the number $t$ where $C_2(\va)$ is the $i$-th least significant bit of $t$,
where $i=\num(\va)$.
The subset-sum instance represented by $C_1$ and $C_2$ is denoted by $\cN(C_1,C_2)$.
We will describe the succinct projection for succinct subset-sum.

Let $C_1(\vu_1,\vv)$ and $C_2(\vu_2)$ be the input where $|\vu_1|=|\vu_2|=n$ and $|\vv|=m$.
We need a few notations.
Let $s_0,\ldots,s_{2^m-1}$ be the numbers represented by $C_1$ and $t$ the number represented by $C_2$.
For a set $X\subseteq [2^m]$, let $T_X = \sum_{i\in X} s_i$.
For $0\leq j \leq 2^m$, let $T_{X,j} = T_{X\cap [j]}$.
Abusing the notation, for $\vb\in\Sigma^m$,
we write $s_{\vb}$ and $T_{X,\vb}$ to denote $s_i$ and $T_{X,i}$, respectively, where $i=\num(\vb)$.
For $\va\in \Sigma^n$, bit-$\va$ means bit-$i$ where $i=\num(\va)$.

We represent a set $X\subseteq [2^m]$ as a function $g:\Sigma^n\times\Sigma^m\to\Sigma^5$
where $g(\va,\vb)=(\alpha,\beta,\gamma,\delta,\epsilon)$ such that:
\begin{itemize}
\item
$\alpha=1$ iff $s_{\vb}\in X$.
\item
$\beta$ is bit-$\va$ in $T_{X,\vb}$.
\item
$\gamma$ is the carry of adding $T_{X,\vb}$ and $s_{\vb}$ up to bit-$(\va-1)$.
\item
$\delta\epsilon=\beta+\gamma+C(\va,\vb)$, i.e., $\epsilon$ is the least significant bit of $\beta+\gamma+C(\va,\vb)$
and $\delta$ is the carry.
\end{itemize}
See the illustration below. 

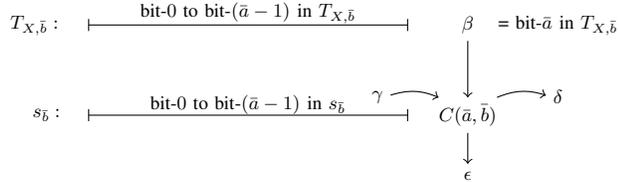
\begin{center}
\begin{tikzpicture}[scale=0.8, every node/.style={scale=0.8}]


\node at (-6.2,1.5) (TXw) {\small$T_{X,\vb}:$};
\draw (-5.3,1.5) -- node[above,yshift=-0.1cm] {\small bit-$0$ to bit-$(\va-1)$ in $T_{X,\vb}$} (0,1.5);
\draw (-5.3,1.4) -- (-5.3,1.6);
\draw (0,1.4) -- (0,1.6);
\node at (1,1.5) (b1) {\small $\beta$};
\node at (2.5,1.5) (x) {\small = bit-$\va$ in $T_{X,\vb}$};

\node at (-6,0) (sw) {\small$s_{\vb}:$};
\draw (-5.3,0) -- node[above,yshift=-0.1cm] {\small bit-$0$ to bit-$(\va-1)$ in $s_{\vb}$} (0,0);
\draw (-5.3,-0.1) -- (-5.3,0.1);
\draw (0,-0.1) -- (0,0.1);
\node at (-.5,0.3) (c) {\small$\gamma$};
\node at (1,0) (Cww) {\small$C(\va,\vb)$};
\node at (2.5,0.3) (d) {\small$\delta$};
\node at (1,-1) (e) {\small$\epsilon$};
\draw [->] (c) to[bend left=20] (Cww);
\draw [->] (Cww) to[bend left=20] (d);
\draw [->] (b1) to (Cww);
\draw [->] (Cww) to (e);

\end{tikzpicture}
\end{center}
\vspace{-0.2cm}
Intuitively, $g(\va,\vb)$ contains the information about 
the additions performed on bit-$\va$ in $s_{\vb}$ (with respect to the set $X$).
In particular, the bits of the number $T_{X}$ are all contained in $g(\va,1^m)$ for every $\va\in \Sigma^n$.
These bits can then be compared to those in $t$ by means of the circuit $C_2$.

Note that for a function $g:\Sigma^n\times\Sigma^m\to\Sigma^5$
to properly represent a number $T_X$, for some $X\subseteq [2^m]$,
it suffices to check the values of $g$ on ``neighbouring'' points in $\Sigma^n\times\Sigma^m$.
More precisely, the following conditions must be satisfied for every $(\va_1,\vb_1),(\va_2,\vb_2)\in \Sigma^n\times\Sigma^m$,
where $g(\va_1,\vb_1)=(\alpha_1,\beta_1,\gamma_1,\delta_1,\epsilon_1)$ and 
$g(\va_2,\vb_2)=(\alpha_2,\beta_2,\gamma_2,\delta_2,\epsilon_2)$.
\begin{enumerate}[(i)]
\item
If $\vb_1=\vb_2$, then $\alpha_1=\alpha_2$.
That is, the value $\alpha_1$ depends only on the index of a number.
\item
If $\alpha_1=0$, then $\gamma_1=\delta_1=0$ and $\beta_1=\epsilon_1$.
\item
If $\alpha_1=1$, then $\gamma_1+C(\va_1,\vb_1)+\beta_1=\delta_1\epsilon_1$.
\item
If $\va_1=0^n$, then $\gamma_1=0$.
\item
If $\va_1=1^n$, then $\delta_1=0$.
\item
If $\vb_1=0^m$, then $\beta_1=\gamma_1=0$.
\item
If $\vb_1=1^m$, then $\epsilon_1=C_2(\va_1)$.
\item
If $\alpha_1=1$ and $\vb_1=\vb_2$ and $\va_2=\va_1+1$, then $\delta_1=\gamma_2$.
\item
If $\alpha_1=1$ and $\vb_2=\vb_1+1$ and $\va_2=\va_1$, then $\epsilon_1=\beta_2$.
\end{enumerate}
Intuitively, (ii) and (iii) state that the values of $(\alpha_1,\beta_1,\gamma_1,\delta_1,\epsilon_1)$
must have their intended meaning, i.e., when $\alpha_1=0$, no addition is performed
and when $\alpha_1=1$, the addition $\gamma_1+C(\va_1,\vb_1)+\beta_1$ is performed and the result is $\delta_1\epsilon_1$.
(iv) states that there is no carry from the previous bit when considering the least significant bit.
(v) states that there shouldn't be any carry after adding the most significant bit (if we want $T_X$ equals $t$).
(vi) states that $T_{X,0}$ must be zero.
(vii) states that bit-$\va$ in $T_{X}$ must equal to bit-$\va$ in $t$.
Finally, (viii) and (ix) state that 
when $(\va_1,\vb_1)$ and $(\va_2,\vb_2)$ are neighbors,
the bits $\beta_1,\gamma_1,\delta_1,\epsilon_1$ and $\beta_2,\gamma_2,\delta_2,\epsilon_2$ 
 must obey their intended meaning. 

Obviously, if $g$ satisfies (i)--(ix),
then it represents a set $X$ such that $T_X=t$.
Conversely, if there is a set $X$ such that $T_X=t$,
then there is a function $g$ that satisfies (i)--(ix).
It is not difficult to design a succinct projection
that constructs a circuit $D$ that describes functions that satisfy (i)-(ix).


\section{Reductions from other $\nexpt$-complete logics}
\label{sec:logic}

In this section we will consider the following fragments of relational first-order logic (with the equality predicate):
\begin{itemize}
\item
The {\em Bernays-Sch\"onfinkel-Ramsey} (BSR) class:
The class of relational $\fo$ sentences of the form:
\begin{align*}
\Psi_1 & := \exists x_1\cdots \exists x_m \ \forall y_1\cdots \forall y_n \ \psi
\end{align*}
where $\psi$ is a quantifier-free formula.
\item
The two-variable logic ($\fotwo$): 
The class of relational $\fo$ sentences using only two variables $x$ and $y$.

The classic result by Scott~\cite{scott} states
that every $\fotwo$ sentence can be transformed in linear time into an equisatisfiable $\fotwo$ sentence
of the form:
\begin{eqnarray*}
\Psi_2 := \forall x \forall y\ \alpha(x,y) \ \wedge \ \bigwedge_{i=1}^m \forall x \exists y \beta_i(x,y)
\end{eqnarray*}
for some $m\geq 1$, where $\alpha(x,y)$ and each $\beta_i(x,y)$ are quantifier free formulas.
\item
The {\em L\"owenheim/monadic} class: 
The class of relational $\fo$ sentences using only unary predicate symbols.
Sentences in this class are also known as monadic sentences.
\end{itemize}
We denote by $\satbsr$, $\satmon$ and $\satfotwo$ 
the corresponding satisfiability problem for each class
and it is well known that all of them are $\nexpt$-complete~\cite{Lewis80,Furer83,gkv,bgg97,llt21}.
The upper bound is usually established by the so called {\em Exponential Size Model} (ESM) property
stated as follows.
\begin{itemize}
\item
If the $\bsr$ sentence $\Psi_1$ is satisfiable,
then it is satisfiable by a model with size at most $m+1$~\cite[Prop.~6.2.17]{bgg97}.
\item
If the $\fotwo$ sentence $\Psi_2$ is satisfiable, 
then it is satisfiable by a model with size $m2^n$,
where $n$ is the number of unary predicates used~\cite{gkv}.
\item
If a L\"owenheim sentence is satisfiable,
then it is satisfiable by a model with size at most $r2^n$,
where $r$ is the quantifier rank and $n$ is the number of unary predicates~\cite[Prop.~6.2.1]{bgg97}.
\end{itemize}

The main idea of the reduction is quite simple.
We will represent the domain of a model with size at most $2^t$
as a subset of $\Sigma^t$
and use a function $f_0:\Sigma^t\to\Sigma$ as the indicator 
whether an element is in the domain.
Every predicate in the input formula 
can be represented as a function $f:\Sigma^{kt}\to\Sigma$ where $k$ is the arity of the predicate.
All these functions can then be encoded appropriately as existential variables in DQBF.
Note that the universal $\fo$ quantifier $\forall x \cdots$ can be encoded as $\forall \vu \ f_0(\vu) \to \cdots$.
The existential $\fo$ quantifier can first be Skolemized which
can then be encoded as existential variables in DQBF.

The rest of this section is organized as follows.
For technical convenience, we first introduce the logic Existential Second-order Quantified Boolean Formula ($\esb$)
-- an alternative, but equivalent formalism of DQBF.
The only difference between $\esb$ and DQBF is the syntax in declaring the function symbol.
Then, we consider the problem that we call {\em Bounded $\fo$ satisfiability},
denoted by $\bsatfo$, which subsumes all $\satbsr$, $\satfotwo$ and $\satmon$
and show how to reduce it to $\satdqbf$.

\paragraph*{The logic $\esb$}

The class $\esb$ is the extension of QBF formulas extended with existential second-order quantifiers.
That is, $\esb$ consists of formulas of the form:
\begin{align*}
\Psi & := 
\exists f_1 \exists f_2 \cdots \exists f_p\ 
Q_{1} v_1 \cdots \ Q_{n} v_{n} \quad \psi
\end{align*}
where each $Q_i \in \{\forall,\exists\}$ and each $f_i$ is a boolean function symbol associated with a fixed arity $\ar(f_i)$.
The formula $\psi$ is a boolean formula using the variables $v_i$'s and $f(\vz)$'s,
where $f\in \{f_1,\ldots,f_p\}$, $|\vz|=\ar(f)$ and $\vz \subseteq \{v_1,\ldots,v_q\}$.
We call each $f(\vz)$ in $\psi$ a {\em function variable}.

The semantics of $\Psi$ is defined naturally.
We say that $\Psi$ is satisfiable, 
if there is an interpretation $F_i:\Sigma^{\ar(f_i)}\to \Sigma$ for each $f_i$
such that $Q_{1} v_1 \cdots \ Q_{n} v_{n} \ \psi$ is a true QBF.
In this case we say that {\em $F_1,\ldots,F_p$ make $\Psi$ true}.
It is not difficult to see that DQBF and $\esb$ can be transformed to each other in linear time 
while preserving satisfiability.


\paragraph*{Bounded $\fo$ satisfiability ($\bsatfo$)}
The problem $\bsatfo$ is defined as:
On input relational $\fo$ sentence $\varphi$
and a positive integer $N$ (in binary),
decide if $\varphi$ has a model with cardinality at most $N$.
It is a folklore that $\bsatfo$ is $\nexpt$-complete.
Note that due to the ESM property,
it is trivial that $\bsatfo$ subsumes 
all of $\satbsr$, $\satfotwo$ and $\satmon$.


\paragraph*{Reduction from $\bsatfo$ to $\satesb$}
Let $\varphi$ and $N$ be the input to $\bsatfo$.
We may assume that $\varphi$ is in the Prenex normal form:
$\varphi := 
Q_1x_1 \cdots Q_n x_n \ \psi$,
where each $Q_i \in \{\forall,\exists\}$
and $\psi$ is quantifier-free formula.
Adding redundant quantifier, if necessary, we may assume that $Q_1$ is $\forall$.
Then, we Skolemize each existential quantifier as follows.
Let $i$ be the minimal index where $Q_i=\exists$.
We rewrite $\varphi$ into:
\begin{align*}
\varphi' & := 
\forall x_1 \cdots \forall x_{i-1}\ Q_{i+1}x_{i+1}\cdots Q_n x_n\ \forall z 
\\ 
& \qquad\qquad z = g(x_1,\ldots,x_{i-1})\ \to \ \psi'
\end{align*}
where $z$ is a fresh variable, $g$ is the Skolem function representing the existentially quantified variable $x_i$
and $\psi'$ is obtained from $\psi$ by replacing every occurrence of $x_i$ with $z$.
Hence,
we may assume that the input sentence $\varphi$ is of form:
\begin{align}
\label{eq:bsatfo}
\varphi & := 
\forall x_1 \cdots \forall x_n \ \psi
\end{align}
where $\psi$ is quantifier-free formula where 
every (Skolem) function symbol $g(x_1,\ldots,x_{i-1})$
only occur in the equality predicate $z = g(x_1,\ldots,x_{i-1})$
and $z$ is one of $x_{i},\ldots,x_n$.

In the following let $g_1,\ldots,g_k$ be the Skolem function symbols in $\psi$
and $P_1,\ldots,P_{\ell}$ be the predicates in $\psi$.
Let $\ar(g_i)$ and $\ar(P_i)$ denote the arity of $g_i$ and $P_i$.
Let $t=\lceil \log N \rceil$.
Construct the following $\esb$ formula:
\begin{align}
\label{eq:esb-logic}
\nonumber
\Phi & := \exists f_0 \ \exists f_{1,1}\cdots \exists f_{1,t} \cdots 
\exists f_{k,1}\cdots \exists f_{k,t}
\ \exists f_{P_1} \cdots \exists f_{P_{\ell}}
\\
& \qquad\qquad \forall \vu_1 \cdots \forall \vu_n \ 
\left(
\begin{array}{l}
\vu_1 = 0^t \to f_0(\vu_1) 
\\
\wedge \ 
\bigwedge_{i=1}^n f_0(\vu_i) \ \to \ \Psi
\end{array}
\right)
\end{align}
where:
\begin{itemize}
\item
The arity of $f_0$ is $t$.
\item 
For every $1\leq i \leq k$,
the arity of $f_{1,1},\ldots,f_{1,t}$ is $t\cdot \ar(g_i)$.
\item 
For every $1\leq i \leq \ell$,
the arity of $f_{P_1},\ldots,f_{P_{\ell}}$ is $t\cdot \ar(P_i)$.
\item 
For every $1\leq i \leq n$, $|\vu_i|=t$.
\end{itemize}
The formula $\Psi$ is obtained from $\psi$ as follows.
\begin{itemize}
\item
Each predicate $P_i(x_{j_1},\ldots,x_{j_m})$ 
is replaced with $f_{P_i}(\vu_{j_1},\ldots,\vu_{j_m})$.
\item 
Each predicate $x_j = g_i(x_{j_1},\ldots,x_{j_m})$
is replaced with $\vu_j = (f_{i,1}(\vu_{j_1},\ldots,\vu_{j_m}), \ldots,f_{i,t}(\vu_{j_1},\ldots,\vu_{j_m}))$
\item 
Each predicate $x_j = x_i$ is replaced with
$\vu_j = \vu_i$.
\end{itemize}
Intuitively, we use $f_0$ as the indicator to determine whether a string in $\Sigma^t$
is an element in the model.
To ensure that the model is not empty, we insist that $0^t$ belongs to the model,
hence, the formula $\vu_1 = 0^t \to f_0(\vu_1)$.
We use the vector of variables $\vu_i$ to represent $x_i$.
For every $1\leq i \leq k$,
the functions $f_{i,1},\ldots,f_{i,t}$ 
represent the bit representation of $g_i(x_{j_1},\ldots,x_{j_m})$.
Finally, for every $1\leq i \leq \ell$,
the function $f_{P_i}$ represents the predicate $P_i$.
Note the part $\bigwedge_{i=1}^n f_0(\vu_i) \ \to \ \Psi$
which means we require $\Psi$ holds only on the vectors $\vu_1,\ldots,\vu_n$ that 
``passes'' the function $f_0$, i.e., 
they are elements of the model.
It is routine to verify that the formula $\varphi$ in Eq.~(\ref{eq:bsatfo})
is satisfiable by a model with cardinality at most $N$
iff the $\esb$ formula $\Phi$ in Eq.~(\ref{eq:esb-logic}) is satisfiable.

\section*{Acknowledgement}

We are very grateful to Jie-Hong Roland Jiang for many fruitful discussions on the preliminary drafts of this work. 
We also thank the anonymous reviewers for their constructive comments.
We acknowledge the generous financial support of 
Taiwan Ministry of Science and Technology under grant no.~109-2221-E-002-143-MY3.

\bibliographystyle{IEEEtran}
\bibliography{aa-bib-succ-rep}

\begin{thebibliography}{10}
\providecommand{\url}[1]{#1}
\csname url@samestyle\endcsname
\providecommand{\newblock}{\relax}
\providecommand{\bibinfo}[2]{#2}
\providecommand{\BIBentrySTDinterwordspacing}{\spaceskip=0pt\relax}
\providecommand{\BIBentryALTinterwordstretchfactor}{4}
\providecommand{\BIBentryALTinterwordspacing}{\spaceskip=\fontdimen2\font plus
\BIBentryALTinterwordstretchfactor\fontdimen3\font minus
  \fontdimen4\font\relax}
\providecommand{\BIBforeignlanguage}[2]{{%
\expandafter\ifx\csname l@#1\endcsname\relax
\typeout{** WARNING: IEEEtran.bst: No hyphenation pattern has been}%
\typeout{** loaded for the language `#1'. Using the pattern for}%
\typeout{** the default language instead.}%
\else
\language=\csname l@#1\endcsname
\fi
#2}}
\providecommand{\BIBdecl}{\relax}
\BIBdecl

\bibitem{sat-handbook}
A.~Biere, M.~Heule, H.~van Maaren, and T.~Walsh, Eds., \emph{Handbook of
  Satisfiability}.\hskip 1em plus 0.5em minus 0.4em\relax {IOS} Press, 2009.

\bibitem{Jiang09}
J.~R. Jiang, ``Quantifier elimination via functional composition,'' in
  \emph{{CAV}}, 2009.

\bibitem{BalabanovJ15}
V.~Balabanov and J.~R. Jiang, ``Reducing satisfiability and reachability to
  {DQBF},'' in \emph{Talk given at {QBF}}, 2015.

\bibitem{SchollB01}
C.~Scholl and B.~Becker, ``Checking equivalence for partial implementations,''
  in \emph{{DAC}}, 2001.

\bibitem{GitinaRSWSB13}
K.~Gitina, S.~Reimer, M.~Sauer, R.~Wimmer, C.~Scholl, and B.~Becker,
  ``Equivalence checking of partial designs using dependency quantified boolean
  formulae,'' in \emph{{ICCD}}, 2013.

\bibitem{BloemKS14}
R.~Bloem, R.~K{\"{o}}nighofer, and M.~Seidl, ``{SAT}-based synthesis methods
  for safety specs,'' in \emph{{VMCAI}}, 2014.

\bibitem{ChatterjeeHOP13}
K.~Chatterjee, T.~Henzinger, J.~Otop, and A.~Pavlogiannis, ``Distributed
  synthesis for {LTL} fragments,'' in \emph{{FMCAD}}, 2013.

\bibitem{KuehlmannPKG02}
A.~Kuehlmann, V.~Paruthi, F.~Krohm, and M.~Ganai, ``Robust boolean reasoning
  for equivalence checking and functional property verification,'' \emph{{IEEE}
  Trans. Comput. Aided Des. Integr. Circuits Syst.}, vol.~21, no.~12, pp.
  1377--1394, 2002.

\bibitem{PetersonR79}
G.~Peterson and J.~Reif, ``Multiple-person alternation,'' in \emph{{FOCS}},
  1979.

\bibitem{BalabanovCJ14}
V.~Balabanov, H.~K. Chiang, and J.~R. Jiang, ``Henkin quantifiers and boolean
  formulae: {A} certification perspective of {DQBF},'' \emph{Theor. Comput.
  Sci.}, vol. 523, pp. 86--100, 2014.

\bibitem{FrohlichKB12}
A.~Fr{\"{o}}hlich, G.~Kov{\'{a}}sznai, and A.~Biere, ``A {DPLL} algorithm for
  solving {DQBF},'' in \emph{{POS-12}, Third Pragmatics of {SAT} workshop},
  2012.

\bibitem{Ge-ErnstSW19}
A.~Ge{-}Ernst, C.~Scholl, and R.~Wimmer, ``Localizing quantifiers for {DQBF},''
  in \emph{FMCAD}, 2019.

\bibitem{KullmannS19}
O.~Kullmann and A.~Shukla, ``Autarkies for {DQCNF},'' in \emph{{FMCAD}}, 2019.

\bibitem{WimmerSB19}
R.~Wimmer, C.~Scholl, and B.~Becker, ``The {(D)QBF} preprocessor hqspre -
  underlying theory and its implementation,'' \emph{J. Satisf. Boolean Model.
  Comput.}, vol.~11, no.~1, pp. 3--52, 2019.

\bibitem{WimmerWSB16}
K.~Wimmer, R.~Wimmer, C.~Scholl, and B.~Becker, ``Skolem functions for
  {DQBF},'' in \emph{{ATVA}}, 2016.

\bibitem{WimmerRM017}
R.~Wimmer, S.~Reimer, P.~Marin, and B.~Becker, ``{HQSpre} -- an effective
  preprocessor for {QBF} and {DQBF},'' in \emph{{TACAS}}, 2017.

\bibitem{Kovasznai16}
G.~Kov{\'{a}}sznai, ``What is the state-of-the-art in {DQBF} solving,'' in
  \emph{Join Conference on Mathematics and Computer Science}, 2016.

\bibitem{SchollW18}
C.~Scholl and R.~Wimmer, ``Dependency quantified boolean formulas: An overview
  of solution methods and applications - extended abstract,'' in \emph{{SAT}},
  2018.

\bibitem{FrohlichKBV14}
A.~Fr{\"{o}}hlich, G.~Kov{\'{a}}sznai, A.~Biere, and H.~Veith, ``{iDQ}:
  Instantiation-based {DQBF} solving,'' in \emph{{POS-14}, Fifth Pragmatics of
  {SAT} workshop}, 2014.

\bibitem{TentrupR19}
L.~Tentrup and M.~Rabe, ``Clausal abstraction for {DQBF},'' in \emph{{SAT}},
  2019.

\bibitem{GitinaWRSSB15}
K.~Gitina, R.~Wimmer, S.~Reimer, M.~Sauer, C.~Scholl, and B.~Becker, ``Solving
  {DQBF} through quantifier elimination,'' in \emph{{DATE}}, 2015.

\bibitem{WimmerKBS017}
R.~Wimmer, A.~Karrenbauer, R.~Becker, C.~Scholl, and B.~Becker, ``From {DQBF}
  to {QBF} by dependency elimination,'' in \emph{{SAT}}, 2017.

\bibitem{SicS21}
J.~S{\'{\i}}c and J.~Strejcek, ``{DQBDD:} an efficient bdd-based {DQBF}
  solver,'' in \emph{{SAT}}, 2021.

\bibitem{avi-succ-83}
H.~Galperin and A.~Wigderson, ``Succinct representations of graphs,''
  \emph{Inf. Control.}, vol.~56, no.~3, pp. 183--198, 1983.

\bibitem{KiniM018}
D.~Kini, U.~Mathur, and M.~Viswanathan, ``Data race detection on compressed
  traces,'' in \emph{{ESEC/SIGSOFT} {FSE}}, 2018.

\bibitem{PavlogiannisSSC20}
A.~Pavlogiannis, N.~Schaumberger, U.~Schmid, and K.~Chatterjee,
  ``Precedence-aware automated competitive analysis of real-time scheduling,''
  \emph{{IEEE} Trans. Comput. Aided Des. Integr. Circuits Syst.}, vol.~39,
  no.~11, pp. 3981--3992, 2020.

\bibitem{yanna-succ-86}
C.~Papadimitriou and M.~Yannakakis, ``A note on succinct representations of
  graphs,'' \emph{Inf. Control.}, vol.~71, no.~3, pp. 181--185, 1986.

\bibitem{Lewis80}
H.~Lewis, ``Complexity results for classes of quantificational formulas,''
  \emph{J. Comput. Syst. Sci.}, vol.~21, no.~3, pp. 317--353, 1980.

\bibitem{Furer83}
M.~F{\"{u}}rer, ``The computational complexity of the unconstrained limited
  domino problem (with implications for logical decision problems),'' in
  \emph{Logic and Machines: Decision Problems and Complexity}, 1983, pp.
  312--319.

\bibitem{gkv}
E.~Gr\"adel, P.~Kolaitis, and M.~Vardi, ``On the decision problem for
  two-variable first-order logic,'' \emph{Bull. Symbolic Logic}, vol.~3, no.~1,
  pp. 53--69, 3 1997.

\bibitem{bgg97}
E.~B{\"{o}}rger, E.~Gr{\"{a}}del, and Y.~Gurevich, \emph{The Classical Decision
  Problem}.\hskip 1em plus 0.5em minus 0.4em\relax Springer, 1997.

\bibitem{llt21}
T.~Lin, C.~Lu, and T.~Tan, ``Towards a more efficient approach for the
  satisfiability of two-variable logic,'' in \emph{LICS}, 2021.

\bibitem{dl-handbook03}
F.~Baader, D.~Calvanese, D.~McGuinness, D.~Nardi, and P.~Patel{-}Schneider,
  Eds., \emph{The Description Logic Handbook: Theory, Implementation, and
  Applications}.\hskip 1em plus 0.5em minus 0.4em\relax Cambridge University
  Press, 2003.

\bibitem{kotek-paper}
S.~Itzhaky, T.~Kotek, N.~Rinetzky, M.~Sagiv, O.~Tamir, H.~Veith, and
  F.~Zuleger, ``On the automated verification of web applications with embedded
  {SQL},'' in \emph{{ICDT}}, 2017, pp. 16:1--16:18.

\bibitem{aut-reasoning}
J.~Robinson and A.~Voronkov, Eds., \emph{Handbook of Automated Reasoning (in 2
  volumes)}.\hskip 1em plus 0.5em minus 0.4em\relax Elsevier and {MIT} Press,
  2001.

\bibitem{tseitin}
G.~Tseitin, ``On the complexity of derivation in propositional calculus,'' in
  \emph{Studies in Constructive Mathematics and Mathematical Logic, Part II},
  1968.

\bibitem{scott}
D.~Scott, ``A decision method for validity of sentences in two variables,''
  \emph{The Journal of Symbolic Logic}, p. 377, 1962.

\bibitem{enderton}
H.~Enderton, \emph{A mathematical introduction to logic}.\hskip 1em plus 0.5em
  minus 0.4em\relax Academic Press, 1972.

\bibitem{libkin}
L.~Libkin, \emph{Elements of Finite Model Theory}.\hskip 1em plus 0.5em minus
  0.4em\relax Springer, 2004.

\end{thebibliography}

\newpage
\onecolumn
\appendix
%
%

\subsection{Proof of Proposition~\ref{prop:dnf}}

The proof essentially uses the same idea as Tseitin's transformation for the quantifier free boolean formulas~\cite{tseitin}.
Let $\Psi$ be a DQBF in the form (\ref{eq:dqbf}) where $\psi$ is in circuit form.
Let $\vx=(x_1,\ldots,x_n)$.
Let $g_1,\ldots,g_k$ be the internal gates in $\psi$ and let $g_k$ be the output gate.
We will represent them with ``fresh'' boolean variables $\vu=(u_1,\ldots,u_k)$.
We will also have fresh variables $\vv=(v_1,\ldots,v_m)$ to represent $y_1,\ldots,y_m$.

Consider the following DQBF:
\begin{align*}
\Psi' & := 
\forall \vx\ \forall \vu\ \forall \vv \
\exists y_1(\vz_1) \cdots \ \exists y_m(\vz_m)
 \Big(\Big(\bigwedge_{i=1}^m v_i\leftrightarrow y_i\Big)\quad \wedge \quad \Big(\bigwedge_{i=1}^k \phi_i\Big)\Big) \ \to \ u_k
\end{align*}
Intuitively, each $\phi_i$ states that ``the value $u_i$ is the value of gate $g_i$.''
More formally, if the gate $g_i$ is an OR-gate with inputs $g_{j_1},\ldots,g_{j_t}$,
then $\phi_i:= u_i \leftrightarrow (u_{j_1}\vee \cdots \vee u_{j_t})$.
If some $g_{j_l}$ is an input gate $x_h$, then replace $u_{j_l}$ with $x_h$.
If it is an input gate $y_h$, then replace $u_{j_l}$ with $v_h$.
Similarly, if $g_i$ is an AND-gate with inputs $g_{j_1},\ldots,g_{j_t}$,
then $\phi_i := u_i \leftrightarrow (u_{j_1}\wedge \cdots \wedge u_{j_t})$.
If $g_i$ is a NOT-gate with input $g_j$,
then $\phi_i := u_i \leftrightarrow \neg u_j$.
It is routine to verify that $\Psi$ and $\Psi'$ are equisatisfiable.

Each $\phi_i$ can be rewritten in CNF.
Thus, the matrix can be rewritten in DNF as follows.
\begin{align*}
&
\bigvee_{i=1}^m (v_i\wedge \neg y_i)\vee (\neg v_i \wedge y_i) \quad\vee \quad\bigvee_{i=1}^k \neg\phi_i \ \vee \ u_k
\end{align*}
Note that each $\phi_i$ uses only variables from $\vx$, $\vu$ and $\vv$.
Thus, the only terms that use the existential variables are $v_i\wedge \neg y_i$ or $\neg v_i \wedge y_i$
which contains at most one existential variables.
This completes the proof of Proposition~\ref{prop:dnf}.

\subsection{A more detailed proof of Theorem~\ref{theo:nexpt-succ-proj}}
\label{app:sec:succ-proj}

We present a more detailed proof of Theorem~\ref{theo:nexpt-succ-proj}.
Let $L\in \nexpt$ and $M$ be a 1-tape NTM that accepts $L$ in time $2^{p(n)}$ for some polynomial $p(n)$.
Let $Q$ and $\Gamma$ be the set of states and the tape alphabet of $M$.
For simplicity, we assume that $M$ accepts only at exactly step $2^{p(n)}$ when the head is in the leftmost cell. 
We also assume that when $M$ makes a non-deterministic move, the head stays still.

Let $\Delta = \Gamma \cup (Q\times \Gamma)$
An accepting run of $M$ on $w$ is represented by a function  
$g:[2^{p(n)}]\times [2^{p(n)}]\to \Delta$,
where $g(i,j)$ is the symbol in cell $i$ in time $j$ in the run.
When $g(i,j)\in Q\times\Gamma$, it indicates the position of the head is in cell $i$ and the state of $M$.

Let the input word $w$ be $b_0b_1\cdots b_{n-1}$.
For a function $g:[2^{p(n)}]\times [2^{p(n)}]\to\Delta$
to represent a correct accepting run of $\cT$ on $w$,
the following must hold for every $i,j,i',j'\in [2^{p(n)}]$.
\begin{enumerate}[(a)]
\item
If $i=j=0$, then $g(i,j)=(q_0,b_0)$.
\item
If $1\leq i \leq n-1$ and $j=0$, then $g(i,j)=b_i$.
\item
If $n \leq i$ and $j=0$, $g(i,0)$ is the blank symbol.
\item
If $j=j'$, then at most one of $g(i,j)$ and $g(i',j')$ is an element of $Q\times \Gamma$.
\item
If $i=0$ and $j=2^{p(n)}-1$,
then $g({i,j})=(q_{acc},\sigma)$ for some tape symbol $\sigma$ where $q_{acc}$ is the accepting state of $\cT$.

\item
If $|i-i'|\leq 1$ and $|j-j'|\leq 1$ (modulo $2^{p(n)}$),
then $g(i,j)$ and $g(i',j')$ must obey the transitions in $\cT$.

For example, if there is a transition $(s,0)\to(s',1,\text{stay})$ and $(s,0)\to(s'',0,\text{stay})$,
we have the following condition.
\begin{itemize}
\item
If $i=i'$ and $j'=j+1$ and $g(i,j)=(s,0)$,
then $g(i,j')=(s',1)$ or $(s'',0)$.
\end{itemize}
Similar condition can be defined for each transition.
\end{enumerate}
Obviously, $[2^{p(n)}]$ can be encoded with $\Sigma^{p(n)}$ and $\Delta$ with $\Sigma^{\ell}$, where $\ell=\log |\Delta|$.

It is routine to design an algorithm that on input $w$, constructs a circuit $C$ (with access to the transitions in $\cM$)
that given $i,j,g(i,j),i',j',g(i',j')$, verifies whether all properties (a)--(f) hold.
That is, $C(i,j,g(i,j),i',j',g(i',j'))=1$ iff $(i,j,g(i,j),i',j',g(i',j'))$ satisfies (a)--(f).
In other words, the function $g$ agrees with $C$ iff it represents a correct accepting run of $M$ on $w$.
Therefore, $w\in L$ iff there is a function that agrees with $C$.


\subsection{More concrete reductions from other $\nexpt$-complete problems}
\label{app:subsec:more-graphs}

\paragraph*{Reduction from succinct independent set}
Succinct independent set is defined as follows.
The input is a circuit $C(\vu,\vv)$ and an integer $k$ (in binary).
The task is to decide if $G(C)$ has an independent set of size $k$.

Let $|\vu|=|\vv|=n$ and $m=\log k$.
As explained in the body, we may assume that $k$ is a power of $2$.
To avoid clutter, we also assume $G(C)$ does not contain self-loop, i.e.,
for every $\va\in \Sigma^n$, $C(\va,\va)=0$.

We represent a set $I\subseteq \Sigma^n$ with size $k$
with an injective function $g:\Sigma^m\to\Sigma^n$ where 
$g(\bin_m(i))$ denotes the $i$-th element in $I$, for every $i\in[k]$.
Now, for $I$ to be an independent set in $G(C)$ with size $k$,
the following property must hold for every $\va_1,\va_2 \in \Sigma^m$.
\begin{itemize}
\item
If $\va_1\neq \va_2$, then $g(\va_1)\neq g(\va_2)$ and $(g(\va_1),g(\va_2))$ is not an edge in $G(C)$.
\end{itemize}
The succinct projection simply outputs a circuit that expresses this property,
i.e., it outputs the following circuit $D(\vx_1,\vy_1,\vx_2,\vy_2)$ where $|\vx_1|=|\vx_2|=m$ and $|\vy_1|=|\vy_2|=n$:
\begin{align*}
& \vx_1\neq \vx_2 \ \to\
\big( \vy_1\neq \vy_2 \ \wedge \ C(\vy_1,\vy_2)=0\big) 
\end{align*}
Obviously, a function $g$ represents an independent set with size $k$ iff it agrees with $D$.
Thus, $G(C)$ has an independent set with size $k$ iff there is a function that agrees with $D$.

\paragraph*{Reduction from succinct subgraph isomorphism}

The input to {\em succinct subgraph isomorphism} is two circuits $C_1(\vu_1,\vv_1)$ and $C_2(\vu_2,\vv_2)$.
The task is to decide if $G(C_1)$ is isomorphic to a subgraph of $G(C_2)$.

Let $|\vu_1|=|\vv_1|=n_1$ and $|\vu_2|=|\vv_2|=n_2$ and $n_1\leq n_2$.
Note that $G(C_1)$ is isomorphic to a subgraph of $G(C_2)$ iff 
there is a function $g:\Sigma^{n_1}\to \Sigma^{n_2}$ such that
the following holds for every $\va_1,\va_2 \in \Sigma^{n_1}$.
\begin{enumerate}[(S1)]
\item
If $\va_1\neq \va_2$, then $g(\va_1)\neq g(\va_2)$.
\item
$C_1(\va_1,\va_2)=C_2(g(\va_1),g(\va_2))$.
\end{enumerate}
The succinct projection outputs the circuit that expresses both (S1) and (S2),
i.e., the circuit $D(\vx_1,\vy_1,\vx_2,\vy_2)$: 
\begin{align*}
& \big(\vx_1\neq\vx_2 \ \to \ \vy_1\neq \vy_2\big)
\ \wedge\
C_1(\vx_1,\vx_2) \ = \ C_2(\vy_1,\vy_2)
\end{align*}
It is immediate that
a function agrees with $D$ iff
it represents an isomorphism from $G(C_1)$ to a subgraph of $G(C_2)$.

\paragraph*{Reduction from succinct vertex cover}
The input to {\em succinct vertex cover} is a circuit $C(\vu,\vv)$ and an integer $k$ (in binary).
The task is to decide if $G(C)$ has a vertex cover of size at most $k$.

Let $|\vu|=|\vv|=n$ and $m=\log (k+1)$.
Here we assume that $G(C)$ is an undirected graph, i.e.,
for every $\va_1,\va_2\in \Sigma^n$, $C(\va_1,\va_2)=C(\va_1,\va_2)$.

We represent a subset $W\subseteq \Sigma^n$ with size at most $k$
with a function $g:\Sigma^n\to\Sigma^m$
where $g$ is injective on the codomain $\{\bin_m(i) \mid i \leq k-1\}$.
That is, if $g(\va_1)=g(\va_2)$ and $\num(g(\va_1)) \leq k-1$, then $\va_1=\va_2$.
We use such $g$ to represent the set $W = \{w \mid \num(g(w))\leq k-1\}$.
For $W$ to be a vertex cover,
the following must hold for every $w_1,w_2 \in \Sigma^n$.
\begin{enumerate}[(V1)]
\item
If $g(\va_1) = g(\va_2)$ and $\num(g(\va_1)) \leq k-1$, then $\va_1=\va_2$. 

That is, $g$ is injective on codomain $\{\bin_m(i) \mid i \leq k-1\}$.
\item
If $C(\va_1, \va_2)=1$, then $\num(g(\va_1))\leq k-1$ or $\num(g(\va_2))\leq k-1$. 

That is, if $(\va_1,\va_2)$ is an edge, then one of them must be in the vertex cover.
\end{enumerate}
This property can be described by 
the following circuit $D(\vx_1,\vy_1,\vx_2,\vy_2)$ where $|\vx_1|=|\vx_2|=n$ and $|\vy_1|=|\vy_2|=m$:
\begin{align*}
&
\Big( \Big(\vv_1 = \vv_2 \ \wedge \ \num(\vv_1) \leq k-1\Big) \ \to\ \vu_1 = \vu_2 \Big)
\ \wedge \
\Big( C(\vu_1, \vu_2)=1\ \to\ \Big(\num(\vv_1) \leq k-1\ \vee\ \num(\vv_2) \leq k-1\Big) \Big)
\end{align*}
Note that a circuit for testing $\num(\vv_1) \leq k-1$ can be constructed in polynomial time in $m$.
It can be easily verified that a function $g$ represents a vertex cover with size at most $k$ in $G(C)$ iff it agrees with $D$.
Thus, $G(C)$ has a vertex cover with size at most $k$ iff there is a function that agrees with $D$.

\paragraph*{Reduction from succinct dominating set}
The input to {\em succinct dominating set} is a circuit $C(\vu,\vv)$ and an integer $k$ (in binary).
The task is to decide if $G(C)$ has a dominating set of size at most $k$.

Let $|\vu|=|\vv|=n$ and $m=\log k$.
We first assume that $k$ is a power of $2$.
We also assume that $G(C)$ is an undirected graph, i.e.,
for every $\va_1,\va_2\in \Sigma^n$, $C(\va_1,\va_2)=C(\va_2,\va_1)$.

We will view a set $W=\{\vc_0,\ldots,\vc_{k-1}\}\subseteq \Sigma^n$
as $W'=\{(0,\vc_0),\ldots,(k-1,\vc_{k-1})\}$, i.e.,
each element in $W'$ is a pair $(i,\vc_i)$ where $i$ is the ``index'' of $\\vc_i$ (in the set $W$).
A dominating set $W'$ in $G(C)$
can be represented as a function $g:\Sigma^n\to\Sigma^m\times \Sigma^n$
that satisfies the following properties.
For every $\va_1,\va_2\in \Sigma^n$, where $g(\va_1)=(\vb_1,\vc_1)$ and $g(\va_2)=(\vb_2,\vc_2)$:
\begin{enumerate}[(D1)]
\item
If $\va_1\in W$, then $\vc_1=\va_1$.
\item
If $\va_1\notin W$, then $\va_1$ is adjacent to $\vc_1$.
\item 
If $\vb_1=\vb_2$, then $\vc_1=\vc_2$.
\end{enumerate}
Intuitively, $g(\va_1)=(\vb_1,\vc_1)$ means that $\vb_1$ is the index of $\va_1$, if $\va_1$ is in the dominating set $W'$,
indicated by the fact that $\vc_1=\va_1$.
If $\va_1$ is not in $W'$, then $\va_1$ must be adjacent to $\vc_1$.
This is what is stated by (D1) and (D2).
Property (D3) simply states that the index of the element in the image must be unique.
The succinct projection outputs the circuit $D(\vx_1,\vy_1,\vz_1,\vx_2,\vy_2,\vz_2)$ 
that expresses (D1)--(D3),
where $|\vx_1|=|\vx_2|=|\vz_1|=|\vz_2|=n$ and $|\vy_1|=|\vy_2|=m$:
\begin{align*}
& \big(\vy_1=\vy_2  \to  \vz_1= \vz_2\big)
\ \wedge\ 
\big(\vx_1\neq \vz_1  \to  C(\vx_1,\vz_1)=1\big)
\end{align*}
It is routine to show that a function $g$ properly represents a dominating set $W'$ with size at most $k$
iff it agrees with $D$.

\paragraph*{Reduction form succinct $\SAT$}

In the standard representation an instance of $\SAT$ is a set of clauses $c_1,\ldots,c_{\ell}$ 
over some variables $v_1,\ldots,v_k$.
Each clause can be encoded as a string in $\Sigma^{\log \ell}$
and each variable a string in $\Sigma^{\log k}$.
Each literal can be encoded as a pair $(b,\va)\in \Sigma\times \Sigma^{\log k}$,
where the bit $b$ represents the ``negativeness'' of the literal, i.e.,
$(0,\va)$ denotes the positive literal $\va$ and $(1,\va)$ the negative literal $\neg \va$.

The succinct representation of a $\SAT$ instance is a circuit $C(t,\vu,\vv)$, 
where $|t|=1$, $|\vu|=\log k$ and $|\vv|=\log \ell$,
such that the following holds for every $\va_1\in \Sigma^{\log k}$ and $\va_2\in \Sigma^{\log \ell}$.
\begin{itemize}
\item
$C(0,\va_1,\va_2) =1$ iff clause $\va_2$ contains literal $\va_1$.
\item
$C(1,\va_1,\va_2) =1$ iff clause $\va_2$ contains literal $\neg \va_1$.
\end{itemize}
Let $F(C)$ denote the boolean formula represented by the circuit $C$.
We define the problem {\em succinct-$\SAT$} as on input circuit $C(t,\vu,\vv)$,
decide whether $F(C)$ has a satisfying assignment.

In the following we will present a succinct projection for succinct-$\SAT$.
Let $C(t,\vu,\vv)$ be an instance of succinct-$\SAT$, where $|\vu|=m$ and $|\vv|=n$.
Note that a satisfying assignment of $F(C)$ can be viewed as 
a function $g:\Sigma^n\to\Sigma\times \Sigma^m$ where for every $\va_1,\va_2\in \Sigma^n$, the following holds.
Let $g(\va_1)=(b_1,\vc_1)$ and $g(\va_2)=(b_2,\vc_2)$.
\begin{enumerate}[(a)]
\item
$C(b_1,\vc_1,\va_1)=1$.
\item 
If $\vc_1=\vc_2$, then $b_1=b_2$.
\end{enumerate}
Intuitively $g(\va_1)=(b_1,\vc_1)$ means the literal $(b,\vc_1)$ makes clause $\va_1$ true.
Condition (a) states that literal $(b,\vc_1)$ is indeed inside clause $\va_1$.
Condition (b) ensures that there is no contradicting literals that are picked to make two different clauses true.
This property can be described by the following circuit $D(\vx_1,y_1,\vz_1,\vx_2,y_2,\vz_2)$
where $|\vx_1|=|\vx_2|=n$, $|y_1|=|y_2|=1$ and $|\vz_1|=|\vz_2|=m$:
\begin{align*}
& C(y_1,\vz_1,\vx_1)=1
\quad\wedge\quad
\big(\vz_1= \vz_2 \ \to \ y_1=y_2\big)
\end{align*} 
It is routine to verify that
a function $g$ agrees with $D$
iff
it represents a satisfying assignment of $F(C)$.


\subsection{Using succinct projections to obtain reductions to other $\nexpt$-complete logics}
\label{app:subsec:logics}

In the main text we have shown how to use succinct projections
to obtain explicit reductions from some concrete $\nexpt$-complete problems/logics to $\satdqbf$.
In this appendix we will show how succinct projections can be used to 
obtain reductions to other logics.

To this end, we introduce the class $\fotwo_1$ which is the class of $\fo$ sentences (without the equality predicate)
of the form:
\begin{eqnarray}
\label{eq:fotwoone}
\Phi
& := &
\forall x \forall y \ \alpha(x,y) \quad\wedge\quad
 \forall x \exists y \ \beta(x,y)
\end{eqnarray}
where $\alpha(x,y)$ and $\beta(x,y)$ are quantifier free formulas
using only unary predicates and without the equality predicate.
Note that $\fotwo_1$ lies in the intersection between $\fotwo$ and the L\"owenheim class.

For technical convenience, we may assume that $\alpha(x,y)$ and each $\beta_i(x,y)$
are written in a {\em circuit form}, i.e.,
a boolean circuit whose input gates are all the possible atomic predicates.
If $S_1,\ldots,S_p$ are all the unary predicates used in the formula,
then $\alpha(x,y)$ and $\beta(x,y)$ are circuits 
with input gates $S_1(x),\ldots,S_p(x),S_1(y),\ldots,S_p(y)$.
Again, such form can be transformed efficiently into
the standard $\fo$ format via Tseitin transformation,
though such transformation requires introducing new binary predicates.

For $\fotwo_1$ and other subclasses of $\fo$,
we adopt standard notations from~\cite{enderton,libkin}.
We use $\cA$ and $\cB$ to denote structures with domain $A$ and $B$
and $P^{\cA}$ denotes the interpretation of a predicate $P$ in $\cA$.
For a unary predicate $S$ and an element $a\in A$, 
$\mychi_S^{\cA}(a)\in \Sigma$ denotes the indicator bit for the membership of $a$ in $S^{\cA}$,
i.e., $\mychi_S^{\cA}(a)=1$ if and only if $a\in S^\cA$.
For unary predicates $S_1,\ldots,S_n$,
$\mychi_{S_1\cdots S_n}^{\cA}(a)$ denotes the string $\mychi_{S_1}^{\cA}(a)\mychi_{S_2}^{\cA}(a)\cdots \mychi_{S_n}^{\cA}(a)$.
When $\cA$ is clear from the context, we omit $\cA$ and write $\mychi_{S}(a)$ and $\mychi_{S_1\cdots S_n}(a)$.

Let $\satfotwoone$ be the problem that given an $\fotwo_1$ sentence,
decide if it is satisfiable.
It is known that $\satfotwoone$ is $\nexpt$-complete~\cite{Furer83,gkv}.

In this appendix we will present the web of reductions as shown in Figure~\ref{fig:web}.
Note that the reductions from $\satbsr$, $\satfotwo$ and $\satmon$ to $\satesb$
have been presented in the main body.
So, what is left is the reduction from $\satesb$ to $\satbsr$
and from $\satesb$ to $\satfotwoone$.

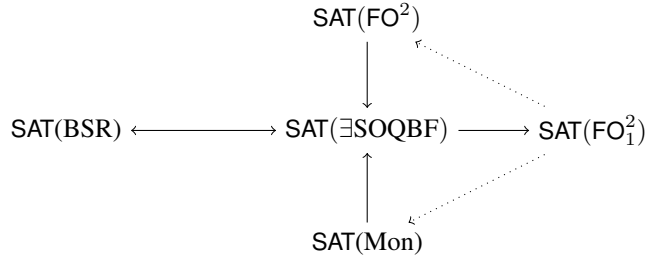
\begin{figure}[t]
\begin{center}
\begin{tikzpicture}

\node at (-4,0) (bsr) {$\satbsr$};
\node at (0,0) (eso) {$\satesb$};
\node at (3,0) (fo21) {$\satfotwoone$};
\node at (0,-1.5) (l) {$\satmon$};
\node at (0,1.5) (fo2) {$\satfotwo$};

\draw[<->] (bsr) to (eso);
\draw[->] (l) to (eso);
\draw[->] (fo2) to (eso);
\draw[->] (eso) to (fo21);
\draw[dotted,->] (fo21) to (fo2);
\draw[dotted,->] (fo21) to (l);

\end{tikzpicture}
\end{center}
\caption{The web of reductions between the considered logics considered in this paper.
Solid arrows indicate the direction of the reduction
and dotted arrows mean ``subsumed by''.}
\label{fig:web}
\end{figure}

We first show how to use succinct projections to obtain reductions from $\nexpt$-complete problems to $\satfotwoone$.

\subsubsection{Reduction to $\satfotwoone$.}
Next, we present the reduction from any language $L \in \nexpt$ to $\satfotwoone$.
Let $L\in \nexpt$ and $\cM$ its succinct projection.

We need a few notations.
Let $R_1,\ldots,R_n,S_1,\ldots,S_m$ be unary predicates.
Define the formulas $\feq_{R_1\cdots R_n}(x,y)$ and $\fsuc_{R_1\cdots R_n}(x,y)$ as follows.
\begin{eqnarray*}
\feq_{R_1\cdots R_n}(x,y) & := & \bigwedge_{i=1}^n R_i(x)\iff R_i(y)
\\
\fsuc_{R_1\cdots R_n}(x,y) & := &  \bigvee_{i=1}^n \Big(
 \neg R_i(x) \wedge R_i(y) \wedge
\bigwedge_{j=1}^{i-1}  \big( R_j(x) \wedge \neg R_j(y)\big)
\wedge
\bigwedge_{j=i+1}^n \big(R_j(x) \iff R_j(y)\big)\Big)
\\
& &
\vee\quad \bigwedge_{i=1}^n \Big(R_i(x) \wedge \neg R_i(y)\Big)
\end{eqnarray*}
$\feq_{S_1\cdots S_m}(x,y)$ and $\fsuc_{S_1\cdots S_m}(x,y)$ are defined analogously.
The meaning of these formulas is as follows.
$\cA,x/a,y/b \models \feq_{R_1\cdots R_n}(x,y)$ iff $\mychi_{R_1\cdots R_n}(a)=\mychi_{R_1\cdots R_n}(b)$
and 
$\cA,x/a,y/b \models \fsuc_{R_1\cdots R_n}(x,y)$ iff $\mychi_{R_1\cdots R_n}(a)+1 = \mychi_{R_1\cdots R_n}(b)$.
Note that $\cA\models\ \forall x \exists y \fsuc_{R_1\cdots R_n}(x,y)$
if and only if for every $w\in \Sigma^n$, there is an element $a$ in $\cA$ such that $\mychi_{R_1\cdots R_n}(a)=w$.

Let $C(\vu_1,\vv_1,\vu_2,\vv_2)$ be a circuit where $|\vu_1|=|\vu_2|=n$ and $|\vv_1|=|\vv_2|=m$.
Let $\vu_i=(u_{i,1},\ldots,u_{i,n})$ and $\vv_i=(v_{i,1},\ldots,v_{i,m})$, for each $i\in \{1,2\}$.

We write $C\big[\vu_1/\vR(x), \vv_1/\vS(x),\vu_2/\vR(y),\vv_2/\vS(y)\big]$ 
to denote the quantifier free $\fo$ formula (in circuit form) obtained from $C$ by 
replacing each $u_{1,i}$ with $R_i(x)$,
each $v_{1,i}$ with $S_i(x)$, each $u_{2,i}$ with $R_i(y)$ and each $v_{2,i}$ with $S_i(y)$.

The reduction from $L$ to $\satfotwoone$ is presented in the following algorithm.

\noindent
\begin{center}
\begin{tabular}{|l|}
 \hline
\multicolumn{1}{|c|}{\bf Algorithm~2: Reducing $L\in \nexpt$ to $\satfotwoone$}
\\ 
\hline
\hline
{\bf Input:} $w\in \Sigma^*$.
\\
{\footnotesize 1:} Run the succinct projection of $L$ on $w$.
\\
{\footnotesize 2:} Let $C(\vx_1,\vy_1,\vx_2,\vy_2)$ be the output circuit where:
\\
\hspace{0.25cm}
$|\vx_1|=|\vx_2|=n$, $|\vy_1|=|\vy_2|=m$, $\vy_1=(y_{1,1},\ldots,y_{1,m})$
 and $\vy_2=(y_{2,1},\ldots,y_{2,m})$.
\\
{\footnotesize 3:} 
Let $R_1,\ldots,R_n,S_1,\ldots,S_m$ be unary predicates.
\\
{\footnotesize 4:} 
Construct the sentence $\Phi:= \forall x \forall y \big(\alpha_1(x,y)\wedge \alpha_2(x,y)\big) \wedge \forall x \exists y\ \beta(x,y)$ where:
\\
\hspace{0.45cm}
--
$\alpha_1(x,y)\ :=\ \ \feq_{R_1\cdots R_n}(x,y)\to \feq_{S_1\cdots S_m}(x,y)$.
\\
\hspace{0.45cm}
--
$\alpha_2(x,y)\ :=\ \ C\big[\vu_1/\vR(x), \vv_1/\vS(x),\vu_2/\vR(y),\vv_2/\vS(y)\big]$.
\\
\hspace{0.45cm}
--
$\beta(x,y)\ \ :=\ \ \fsuc_{R_1\cdots R_n}(x,y)$.
\\
{\footnotesize 5:} 
Output $\Phi$.
\\
\hline
\end{tabular}
\end{center}

To prove the correctness of Algorithm~2,
we need a few terminology.
Let $g:\Sigma^n\to\Sigma^m$ be an arbitrary function.
Let $\cA$ be a structure with unary predicates $R_1,\ldots,R_n,S_1,\ldots,S_m$.
We say that {\em $\cA$ encodes $g$}, if the following holds.
\begin{itemize}
\item
For every $w\in \Sigma^n$, there is $a\in A$ such that $\mychi_{R_1\cdots R_nS_1\cdots S_m}(a)=wg(w)$.
\item
Conversely, for every $a\in A$, there is $w\in\Sigma^n$ such that  $\mychi_{R_1\cdots R_nS_1\cdots S_m}(a)=wg(w)$.
\end{itemize}
Intuitively, the function $g$ is represented by a structure where each $wg(w)$ is encoded
by the membership of the elements in $R_1,\ldots,R_n,S_1,\ldots,S_m$.
Note that if $\cA\models \forall x \forall y\ \alpha_1(x,y) \ \wedge \ \forall x \exists y \ \beta(x,y)$,
then $\cA$ encodes some function $g:\Sigma^n\to\Sigma^m$.

Let $C$ and $n$ and $m$ be as in Steps 1 and 2 in Algorithm~2.
To prove the correctness of Algorithm~\ref{alg:red-to-fo2}, we show the following.
\begin{enumerate}[(a)]
\item
For every function $g$ that agrees with $C$, there is $\cA\models \Phi$ that encodes $g$.
\item
Conversely, for every $\cA\models \Phi$,
there is a function $g$ that agrees with $C$ such that $\cA$ encodes $g$.
\end{enumerate}
From (a) and (b), it follows immediately that $\Phi$ is satisfiable iff $w\in L$.

To prove (a), let $g$ be a function that agrees with $C$.
Let $\cA$ be a structure that encodes $g$.
Thus, $\cA\models \forall x \forall y\ \alpha_1(x,y) \ \wedge \ \forall x \exists y \ \beta(x,y)$.
Since $g$ agrees with $C$, it follows also that $\cA\models \forall x \forall y \ \alpha_2(x,y)$.
Therefore, $\cA\models \Phi$.

To prove (b), let $\cA\models \Phi$.
Let $g:\Sigma^n\to\Sigma^m$ be the function encoded by $\cA$.
Such function $g$ exists since $\cA\models \forall x \forall y\ \alpha_1(x,y) \wedge \forall x \exists y \ \beta(x,y)$.
Moreover, since $\cA\models \forall x \forall y \ \alpha_2(x,y)$,
it follows that $g$ agrees with $C$.

%
%

\subsubsection{Reduction from $\satesb$ to $\satbsr$}

Note that by standard Skolemization, every $\esb$ formula can be transformed into an equivalent formula 
in the normal form: 
$$
\exists f \forall v_1 \forall v_2  \cdots \forall v_{n} \ \psi
$$
That is, there is only one second order quantifiers and all the first-order quantifiers are universal.

Let $\Psi:=\exists f  \forall \vu \ \psi$ be the input $\esb$ formula,
where $\vu= (u_1, \ldots,u_n)$.

We construct a BSR sentence of the form:
\begin{eqnarray*}
\Phi & := & \exists x_0\exists x_1 \forall y_1\cdots \forall y_n \quad x_0\neq x_1 \ \wedge\
 \bigwedge_{i=1}^n \big(y_i = x_0\ \vee \ y_i=x_1\big) \quad \wedge \quad \varphi
\end{eqnarray*}
where  $\varphi$ is obtained from $\psi$ by replacing each 
$u_i$ with $y_i=x_1$, each $\neg u_i$ with $y_i=x_0$
and each function variable $f(\vz)$ with $P(\vz[u_1/y_1,\ldots,u_n/y_n])$.
Here $\vz[u_1/y_1,\ldots,u_n/y_n]$ denotes the vector of variable
where each $u_i$ is replaced with $y_i$.
Recall that $\vz \subseteq \{u_1,\ldots,u_n\}$ for each function variable $f(\vz)$.

Intuitively, the boolean algebra with an interpretation $F$ that makes $\Psi$ true
is viewed as a model with two elements $x_0$ and $x_1$, where $F$ is represented a predicate.
It is not difficult to show that $\Psi$ is a true formula iff
$\Phi$ is satisfiable.

\subsubsection{Reduction from $\satesb$ to $\satfotwoone$}
We will present a succinct projection for $\satesb$.
Let $\Psi:=\exists f  \forall \vu \ \psi$ be an $\esb$ formula,
where $\vu= (u_1,\ldots,u_n)$ and $\ar(f)=k$.
Let there be $m$ function variable $f(\vz_1),\ldots,f(\vz_m)$ in $\psi$.

We first introduce a few notations.
For a vector of variables $\vz$ where $\vz\subseteq \{u_1,\ldots,u_n\}$,
and $w=b_1\cdots b_n\in \Sigma^n$, we write $\vz[\vu/w]$ to denote the string of length $|\vz|$
obtained by replacing each $u_i$ with $b_i$.
Analogously, for a vector $\vv=(v_1,\ldots,v_n)$, we write $\vz[\vu/\vv]$
to denote the vector of variables obtained by replacing each $u_i$ with $v_i$.
For a function $F:\Sigma^k\to\Sigma$,
we define a function $g_F:\Sigma^n\to\Sigma^m$,
where $g(w)=F(\vz_1[\vu/w])\cdots F(\vz_m[\vu/w])$.
Note that $g(w)$ is the string $F(\vz_1)\cdots F(\vz_m)$
when the variables $\vu$ are assigned with $w$.

We now describe the succinct projection for $\satesb$.
It constructs the following circuit $D(\vv_1,\vv_1',\vv_2,\vv_2')$, 
where $|\vv_1|=|\vv_2|=n$ and $|\vv_1'|=|\vv_2'|=m$.
Let $\vv_1' =(v_{1,1}',\ldots,v_{1,m}')$ and $\vv_2' =(v_{2,1}',\ldots,v_{2,m}')$.
\begin{eqnarray*}
D(\vv_1,\vv_1',\vv_2,\vv_2') & := & \psi'(\vv_1,\vv_1') \ \wedge \
\bigwedge_{i=1}^{\ell}\bigwedge_{j=1}^{\ell}
\vz_i[\vu/\vv_1] = \vz_j[\vu/\vv_2] \ \to \ v_{1,i}'=v_{2,j}'
\end{eqnarray*}
where $\psi'(\vv_1,\vv_1')$ is the formula obtained from $\psi$
by replacing every function variable $f(\vz_i)$ with the $v_{1,i}'$.
Obviously, $D$ can be constructed in polynomial time.

The correctness of the succinct projection follows from the following two statements.
\begin{itemize}
\item
For every function $F:\Sigma^{\ell}\to\Sigma$ that makes $\Phi$ true,
the function $g_F$ agrees with $D$.
\item
Conversely, for every function $g$ that agrees with $D$,
there is a function $F$ that makes $\Phi$ true such that $g_F=g$.
\end{itemize}
The proof is routine and hence, omitted.

\end{document}